\documentclass[letterpaper, 10 pt, conference]{ieeeconf}  

\IEEEoverridecommandlockouts                              

\usepackage[utf8]{inputenc}
\usepackage[T1]{fontenc}

\usepackage{cite}
\usepackage{amsmath, amssymb, amsfonts}
\usepackage{mathtools}

\usepackage{algorithm}
\usepackage{algpseudocode}

\usepackage{graphicx}
\usepackage{subcaption}
\usepackage{xcolor}
\usepackage{booktabs}
\usepackage{siunitx}

\newif\ifreviewcolor
\reviewcolorfalse 
\ifreviewcolor
  \colorlet{revcolor}{blue}
\else
  \colorlet{revcolor}{black}
\fi

\newif\ifpicolor
\picolorfalse 
\ifpicolor
  \colorlet{rev2blue}{blue}
\else
  \colorlet{rev2blue}{black}
\fi

\usepackage{url}
\usepackage{hyperref}

\usepackage{dblfloatfix}

\usepackage{etoolbox}
\patchcmd{\thebibliography}
  {\vskip 0.3\baselineskip plus 0.1\baselineskip minus 0.1\baselineskip}
  {\vskip -5pt plus 0.1\baselineskip minus 0.1\baselineskip}
  {\typeout{BIBPATCH-SUCCESS}}{\typeout{BIBPATCH-FAILED}}

\newtheorem{rem}{Remark}
\newtheorem{prob}{Problem}
\newtheorem{prop}{Proposition}
\newtheorem{defi}{Definition}

\newtheorem{thm}{Theorem}
\newtheorem{exam}{Example}
\newtheorem{cor}{Corollary}
\newtheorem{asm}{Assumption}

\newcommand{\Dc}{\mathcal{D}}

\newcommand{\Gc}{\mathcal{G}}

\newcommand{\Rb}{\mathbb{R}}
\newcommand{\Tc}{\mathcal{T}}

\newcommand{\Ad}{\mathrm{Ad}}

\renewcommand{\paragraph}[1]{}

\hypersetup{
    colorlinks=true,
    linkcolor=black,
    citecolor=black,
    urlcolor=blue
}

\def\BibTeX{{\rm B\kern-.05em{\sc i\kern-.025em b}\kern-.08em
    T\kern-.1667em\lower.7ex\hbox{E}\kern-.125emX}}

\title{\LARGE \bf
How Sensor Attacks Transfer Across Lie Groups
}

\author{Rijad Alisic and Saurabh Amin*
\thanks{*Rijad Alisic and Saurabh Amin are with the Laboratory for Information \& Decision Systems at Massachusetts Institute \textcolor{rev2blue}{of} Technology, USA. \mbox{(e-mail: alisic@mit.edu, amins@mit.edu)}. This work was supported by the Knut and Alice Wallenberg Foundation's Postdoctoral Fellowship at MIT.}
}

\begin{document}

\maketitle
\thispagestyle{empty}
\pagestyle{empty}

\begin{abstract}
    \textcolor{rev2blue}{Sensor spoofing analysis in cyber-physical systems is predominantly confined to linear state spaces, where an attack's persistence implies the existence of an integrator somewhere in the loop. We extend this insight to a commutativity condition on Lie groups, where noncommutative dynamics can distort sensor attacks, exposing nominally stealthy attacks by complex maneuvers. We present a geometric framework characterizing when a sensor attack can transfer across operating conditions, preserving both its physical impact and stealthiness. We prove that successful transfer requires the attack to commute with the flow (a Lie bracket condition), isolating transferable attacks to an invariant subspace. For small deviations from transferable attacks, our decomposition theorem reveals a fundamental asymmetry: the flow's Adjoint action distorts the physical impact of the bracket-violating component. Furthermore, even if the attack's impact isn't distorted, the subsequent residual could be. Finally, we demonstrate how turning maneuvers on a Dubins unicycle collapse the transferable subspace to a single direction, verifying that imperfect attacks remain within theoretical detection bounds.}
\end{abstract}
\vspace{-6pt}
\section{Introduction}
\vspace{-5pt}

Sensor spoofing is a well-documented threat to cyber-physical systems, increasingly exacerbated by their reliance on shared communication infrastructure. Documented incidents range from GPS spoofing of autonomous vehicles and drones \cite{tippenhauer2011requirements, kerns2014unmanned}, to false data injection disrupting platoon kinematics \cite{lyons2026distributed} and the dynamical stability of critical infrastructure such as water distribution networks \cite{amin2013cyber}. \textcolor{rev2blue}{The theoretical foundations for such attacks in linear systems have been thoroughly developed \cite{mo2009replay, pasqualetti2013attack}. However, existing literature focuses on detectability~\cite{chong2019tutorial} since an attack's effect on a system is relatively predictable because linear dynamics apply the same map for every state, leaving transferability unaddressed. Nonlinear systems lack this uniformity: the local map varies with the operating condition, so an attacker must typically know the exact state.}

Many safety-critical systems, however, have nontrivial geometric state spaces: vehicle poses on $SE(2)$ or $SE(3)$, robot orientations on $SO(3)$, and nonholonomic vehicles on the Heisenberg group \cite{bloch2015nonholonomic}. Stealthy attack generation has been extended to general nonlinear systems using geometric control~\cite{zhang2022stealthy}, but these approaches miss the algebraic symmetries specific to Lie groups. \textcolor{rev2blue}{The authors in~\cite{kwon2024perfectly} invoke Lie-group automorphisms for undetectable attacks on a manipulator, but its structure, an abelian torus, is commutative by definition. When operations genuinely do not commute, however,} nonlinear motion can rotate the attack relative to its state: an attack stealthy during straight-line driving may become detectable during a turn. Transferability thus becomes a geometric constraint that must be satisfied jointly with stealthiness. \textcolor{rev2blue}{We develop theory for systems on Lie groups with left-invariant dynamics} \cite{murray1994mathematical}, the same algebraic structure exploited by {\color{revcolor}equivariant observers such as the invariant extended Kalman filter} \cite{barrau2017invariant, mahony2022observer} for robust filtering, \textcolor{rev2blue}{but} exploited here instead by an attacker.

Concurrently, behavioral frameworks like DeePC \cite{coulson2019deepc,berberich2020nonlinear} demonstrate that observed trajectories can bypass explicit modeling for control, inspiring a data-driven attack perspective \cite{alisic2021ecc, taheri2021data}. {\color{rev2blue}We implicitly assume an attacker capable of generating such data-driven replay attacks; while this is trivial in linear spaces, extending it to non-Euclidean spaces remains unexplored, since general nonlinear dynamics can distort injected signals.

This paper introduces a unified algebraic framework to characterize transferability.} We answer a fundamental question: what algebraic conditions guarantee that an attack learned in one operating condition retains identical dynamical impact in another? We call such attacks \emph{transferable}; if they furthermore evade detection, they are \emph{stealthily transferable}. By reducing trajectory-level security to finite-dimensional algebraic checks, our contributions are:
\begin{enumerate}
    \item \textcolor{rev2blue}{A formalization of transferability across the state space, alongside a new notion of stealth on Lie groups.}
    \item A Lie-algebraic characterization of transferability, showing both the attack signal and \textcolor{rev2blue}{detector residual} must independently reside in the invariant subspace $\Delta_{f_e} = \ker([f_e, \, \cdot])$.
    \item A decomposition theorem quantifying how bracket-violating residuals alter dynamical impact and stealth through the adjoint action.
\end{enumerate}

\textcolor{rev2blue}{This decomposition reveals that dynamical impact and stealth both depend on $\Delta_{f_e}$, but respond differently. The bracket-violating component warps dynamical impact through the flow's Adjoint action, acting as a kinematic lever-arm. The detector's residual, by contrast, accumulates linearly and is only amplified by the flow's Adjoint action at the following timestep. Consequently, small attacker errors can severely impact the dynamics while stealth deterioration stays bounded. Complex maneuvers can collapse this invariant subspace, giving defenders a principled lever to reveal the accumulated residual induced by attacks.}

The paper is structured as follows. Section~\ref{sec:problem} defines the problem setup. Section~\ref{sec:transfer} derives the Lie algebraic subspace determining invariant transfer. Section~\ref{sec:main_results} analyzes deviations from these nominal settings using our decomposition theorem. Section~\ref{sec:examples} demonstrates the theory on a Dubins unicycle model, and Section~\ref{sec:conclusions} concludes the paper.

\section{Problem Setup}\label{sec:problem}
\vspace{-5pt}
The state of a general dynamical system naturally evolves on a \emph{manifold}. To analyze continuous dynamics, we require smooth manifolds, which allow us to define derivatives and tangent velocities. In this work, we focus on a subclass of smooth manifolds with an algebraic structure: \emph{Lie groups}.

\vspace{-5pt}
\subsection{Lie Groups}
\vspace{-5pt}
Let us introduce some important tools for our analysis.
\begin{defi}
    A \emph{Lie group} $\Gc$ is a smooth manifold equipped with smooth group operations: multiplication $\mu(g,h) = g \cdot h$, identity $e \in \Gc$, and inverse $g^{-1}$.
\end{defi}

For a smooth curve $\gamma: \Rb \to \Gc$ with $\gamma(0) = g$, the \emph{tangent space} at $g$ is $\Tc_g\Gc = \{ \dot{\gamma}(0) \mid \gamma(0) = g \}$.

General manifolds rely on local operations like following the flow of a vector field to define movement across the space. Lie groups, however, use a global algebraic rule to define a \emph{translation operator} $T_g: \Gc \to \Gc$, shifting any $h \in \Gc$ by $g \in \Gc$ via group multiplication, \textcolor{rev2blue}{$\mu(g,h) = gh$}.
{\color{revcolor}For $g\in\Gc$, write $L_g:\Gc\to\Gc,\ \mu(g,\cdot)$ for left translation and $R_g:\Gc\to\Gc,\ \mu(\cdot,g)$ for right translation; $T$ denotes whichever of the two matches the system's invariance class.}
Applying the differential to these translations yields the \emph{induced mapping} (or pushforward) $T_{g*}: \Tc_h \Gc \to \Tc_{T_g h} \Gc$.

Since a tangent vector $\dot{\gamma}(0)$ lives only in the local space $\Tc_g\Gc$, comparing dynamics at different points on the manifold is difficult. To compare vectors globally, we push them to the identity $e$ via $T_{g^{-1}*}\dot{\gamma}(0) \in \Tc_e\Gc$, and denote this tangent space at the identity as $\mathfrak{g} = \Tc_e\Gc$. Equipping this vector space with the \emph{Lie bracket} $[\cdot,\cdot]: \mathfrak{g}\times\mathfrak{g}\to\mathfrak{g}$ forms the \emph{Lie algebra}, equipped with a norm $\Vert \cdot \Vert_{\mathfrak{g}}$.

The Lie algebra is useful due to the existence and uniqueness of its one-parameter subgroup: for $\xi \in \mathfrak g$, the ODE $\dot \gamma(t) = T_{\gamma(t)*} \xi$, $\gamma(0) = e$, has a unique solution $\gamma(t) = \exp (t \xi)$, whose unit displacement is $\gamma(1) = \exp(\xi)$.

These curves can be translated to $g$ through $T_g \exp(t\xi) \in \Gc$. To understand how this translation affects the curve, we use the conjugate map $c^T_g: \Gc \to \Gc$, defined for any $h \in \Gc$ by $c^T_g(h) = T_g T_h T_{g^{-1}} e$. The \emph{adjoint action} $\Ad^T_g: \mathfrak{g} \to \mathfrak{g}$ is then
\begin{equation}\label{eq:Ad}
    \Ad^T_g\,\xi \;:=\;
    \frac{\mathrm{d}}{\mathrm{d}t}\bigg|_{t=0}
    c^T_g\bigl(\exp(t\xi)\bigr)
    \;\in\;\mathfrak{g},
\end{equation}
with induced operator norm $\|\Ad^T_g\| := \sup_{\|\xi\|_{\mathfrak{g}}=1} \|\Ad^T_g \xi\|_{\mathfrak{g}}$.

\vspace{-5pt}
\subsection{Dynamical Systems on Lie Groups}
\vspace{-5pt}
Consider a curve $x(t) \in \Gc$ arising from the dynamics
\begin{equation}
    \dot x(t) = f(x, u), \quad x(0) = x_0.
\end{equation}
We assume the dynamics are $T$-invariant, meaning $f(x,u) = T_{x*} f(e,u)$, {\color{revcolor}which grants completeness: the one-parameter subgroup $t\mapsto\exp(tf_e(u))$ is defined for all $t\in\mathbb{R}$~\cite{hall2015lie}.} Under zero-order hold conditions where the input $u$ is held \emph{constant}, we define $f_e(u) \coloneqq f(e,u) \in \mathfrak{g}$. We write the exact solution over this interval as
\begin{equation}
    x(t) = T_{\exp\bigl(t\,f_e(u)\bigr)}x(0).
\end{equation}
This integration underlies our discrete-time model $x_k = x(t_k)$, sampled at discrete intervals. We write $f_k \coloneqq f_e(u_k)$ for the realized dynamics at step $k$, leaving the dependence on $u_k$ implicit. Specifically, we consider a deterministic observation map $h\colon\Gc\to\mathbb{R}^m$, $z_k = h(x_k)$.
These measurements $z_k$ are subsequently used to determine the next control input $u_k$, or fed into a detector to maintain a state estimate $\hat{x}_k$.

The measurements are the attacker's point of entry. In particular, the attacker modifies the measurements through the following mechanism.
\begin{defi}\label{def:obs_action}
    For $a\in\Gc$, the map
    {\color{revcolor}$\tilde{T}_a\colon\mathbb{R}^m\to\mathbb{R}^m$} defined by
    \begin{equation}
        \tilde{T}_ah(x) \;=\; h(T_ax),
    \end{equation}
    is called an \emph{observation action}.
\end{defi}
\textcolor{rev2blue}{While a spoofing mechanism may be simple to apply in isolation, for instance by adding fixed biases to sensor readings, it is only \emph{physically consistent} spoofing mechanisms that are observation actions, as shown in the next example.}
\begin{exam}
    {\color{revcolor}The induced map $\tilde{T}_a$ models changes to sensor measurements. While isolated GPS spoofing ($x, y$ translations) or isolated LIDAR spoofing (forward $f$, lateral $l$ translations) are trivially observation actions, mixed sensor suites require geometric coordination. For $h(\boldsymbol{x}) = \begin{bmatrix} x & y & f & l \end{bmatrix}^\top$, a spoofing attack $\tilde{T}_a h(\boldsymbol{x}) = h(\boldsymbol{x}) + \begin{bmatrix} a_x & a_y & a_f & a_l \end{bmatrix}^\top$ is only a valid observation action if $a_x = a_f \cos \theta - a_l \sin \theta$ and $a_y = a_f \sin \theta + a_l \cos \theta$, ensuring $\tilde{T}_a h(\boldsymbol{x}) = h(T_a \boldsymbol{x})$.}
\end{exam}
{\color{revcolor}As demonstrated, physically realizing an observation action requires state-dependent knowledge, which is prohibitive to compute analytically in modern systems featuring thousands of sensors with disparate, hard-to-model interdependencies.}

To keep matters general, we simply assume that an approximate observation action is learned from data.
\begin{defi}\label{def:dataset}
    Let $\{x^D_k\}_{k=0}^{N-1}$ be a nominal trajectory starting at $x^D_0\in\Gc$. The \emph{attack dataset} $\Dc^{a}_{x_0^D}$ is a collection of $M$ experiments
    \begin{equation}
        \Dc^{a}_{x_0^D}
        \;=\; \Bigl\{\!
              \bigl\{
                  (z^D_k,\; z_k^{a,(i)})
              \bigr\}_{k=0}^{N-1}
          \!\Bigr\}_{i=1}^{M},
    \end{equation}
    where $z^D_k = h(x^D_k)$ and $z_k^{a,(i)} = h(T_{\exp(\xi_k^{a,(i)})}x^D_k)$, with attacks modeled via one-parameter subgroups $\xi_k^{a,(i)}\in\mathfrak{g}$.
\end{defi}
The attacker leverages this dataset to construct an approximation of the observation action $\tilde{T}_a$, which can then be deployed to manipulate a victim's sensor stream.


{\color{revcolor}The alternative trajectory that is induced by the sensor attack in Definition~\ref{def:obs_action} furthermore interacts with the dynamics. For linear systems, it looks as follows: denote the spoofed state as $x^a_k = x_k + \xi_k$, where $\xi_k$ is the induced difference between the trajectories at time step $k$. In the next time step, the difference between the nominal and attacked state due to this displacement becomes $x^a_{k+1}-x_{k+1} = A(x_{k}^a-x_{k}) = A \xi_k$. The system's inherent dynamics change the displacement. We generalize this notion of dynamical impact to systems on a Lie group next.}

{\color{revcolor}Consider a nominal trajectory $x_0, x_1, \dots$ with $x_{k+1} = T_{\exp(\Delta t\, f_k)}x_k$, where we let $\Delta t = 1$ for brevity, and an attacked trajectory $x_k^a = T_{\exp(\xi_k)}x_k$ for displacements $\xi_k \in \mathfrak{g}$. 
\begin{defi}\label{def:dyn_impact}The \emph{dynamical impact} is
    \begin{equation}\label{eq:dyn_impact}
        d_k := c^T_{\exp(f_k)} \exp(\xi_k) \in \Gc,
    \end{equation}
    where $c^T_g$ is the conjugation operator from~\eqref{eq:Ad}.
\end{defi}
Two attack sequences $\{\xi_k^1\}$ and $\{\xi_k^2\}$ have the same dynamical impact if $d_k^1 = d_k^2$ for all $k$.}

{\color{revcolor}Let us exemplify this with a linear system, where $T=L$, $\exp(f_k)=\begin{pmatrix}A&Bu_k\\0&1\end{pmatrix}$, and $\exp(\xi_k)=\begin{pmatrix}1&\xi_k\\0&1\end{pmatrix}$. Then, we have that $d_k=\begin{pmatrix}1&A\xi_k\\0&1\end{pmatrix}$, matching the computation above.
\begin{defi}\label{def:transferable}
    \textcolor{rev2blue}{A} displacement $\xi_k\in\mathfrak g$ is \emph{transferable} if $$d_k=\exp(\xi_k).$$
\end{defi}
\textcolor{rev2blue}{A dynamical impact that matches the displacement that produced it is transferable. For the linear example, this reduces to $A\xi_k=\xi_k$: the attack survives the dynamics unchanged.}}

To model detection, we define the \emph{observation displacement} $\eta_k \in \mathfrak{g}$ between the true state $x_k$ and the detector's estimate $\hat{x}_k$ as:
\vspace{-7pt}
\begin{equation}\label{eq:eta_def}
    \hat{x}_k = T_{\exp(\eta_k)} x_k. \vspace{-5pt}
\end{equation}
To preserve the geometric structure, this displacement updates multiplicatively:
\begin{equation}\label{eq:eta_update}
    \exp(\eta_k) = T_{\exp \big( K(I_k^{\textcolor{rev2blue}{\rho}}) \big)} \exp\bigl(\Ad_{\exp(f_{k-1})}^T \eta_{k-1}\bigr)
\end{equation}
where $K : \mathbb{R}^m \to \mathfrak{g}$ is an update function on the \emph{innovation} $I_k^{\textcolor{rev2blue}{\rho}}$ \textcolor{rev2blue}{(defined below)}, {\color{revcolor} with the property $K(0) = 0$}.

Under an equivariant observer framework, the innovation uses a geometrically invariant state error. Given the predicted state $\hat{x}_{k|k-1} = T_{\exp(f_{k-1})}\hat{x}_{k-1}$ and the attacked measurement $\tilde{T}_{a_k}h(x_k) = h(T_{a_k} x_k)$, \textcolor{revcolor}{define}
\begin{equation}\label{eq:rho_def}
    \textcolor{revcolor}{\rho_k \coloneqq T_{a_k} T_{x_k} \hat{x}_{k|k-1}^{-1}, \qquad \nu_k \coloneqq c^T_{\hat{x}_{k|k-1}^{-1}}(\rho_k).}
\end{equation}
\textcolor{rev2blue}{The invariant innovation $I_k^\rho \coloneqq h(\rho_k) - h(e)$ is what drives the observer via $\exp(K(I_k^\rho))$, since, opposed to a naive $h(T_{a_k}x_k)-h(\hat x_k)$, it inherently cancels the global reference frame. The detector, however, can also track $I_k^\nu \coloneqq h(\nu_k) - h(e)$ for detection.} We can now define stealth formally.
\begin{defi}\label{def:stealth}
    Let $a_k \in \Gc$ be an attack translation at time $k$. Under the equivariant observer framework, the attack $a_k$ is:
    \begin{enumerate}
        \item \emph{undetectable} if \textcolor{rev2blue}{$I_k^\rho = 0$ and $I_k^\nu = 0$}.
        \item \emph{$\tau$-stealthy} if \textcolor{rev2blue}{$\|I_k^\rho\|\leq\tau$ and $\|I_k^\nu\|\leq\tau$}.
    \end{enumerate}
\end{defi}
\textcolor{revcolor}{For LTI systems, let $\exp(\xi_k)=\begin{pmatrix}\alpha_k&\beta_k\\0&1\end{pmatrix}$ denote a joint multiplicative ($\alpha_k$) and additive ($\beta_k$) attack, and $\hat x_{k|k-1}=\begin{pmatrix}A^k&\hat t_k\\0&1\end{pmatrix}$, with measurement  $h(\hat x_{k|k-1})=C\hat t_k$. Then
\begin{equation*}
    h(\nu_k)-h(e) = CA^{-k}((\alpha_k-I)\hat t_k + \alpha_kr_k+\beta_k),
\end{equation*}
where $h(\rho_k)-h(e)=C\big(\alpha_kr_k+\beta_k\big)$. Setting $\alpha_k=I,\ \beta_k=-r_k$ yields $h(\rho_k)=h(\nu_k)=h(e)$ regardless of $C,A$: we see that some attacks make both innovations zero simultaneously.}

{\color{rev2blue}
\begin{asm}\label{asm:additive_attack}
    We assume that $\|I_k^\rho\| \leq \tau \Rightarrow \|I_k^\nu\| \leq \tau$.
\end{asm}
}

Given these definitions, we are ready to state the problem.
\begin{prob}\label{prob:dyn_cons_traj_att}
    When does a $\tau$-stealthy observation action $\tilde{T}_a$, learned from $\Dc^{a}_{x_0^D}$, transfer to an arbitrary victim state $x \neq x_0^D$ while preserving its dynamical impact and $\tau$-stealth?
\end{prob}

To address Problem~\ref{prob:dyn_cons_traj_att}, the remainder of this paper takes the following assumption:
{\color{revcolor}
\begin{asm}\label{asm:right_translation}
    Our system has \emph{left-invariant dynamics} and attacks are \emph{right translations}
    \begin{equation}
        T_{\exp(\xi)}x=R_{\exp(\xi)}x = x \exp(\xi).
    \end{equation}
\end{asm}
}
We adopt this convention because it naturally captures body-frame sensor displacements in standard robotics models (e.g., $SE(2)$ and $SE(3)$). Equivalent results for right-invariant dynamics or left-translation attacks follow directly by symmetry and are omitted for brevity.

\vspace{-5pt}
\section{Transfer Conditions}\label{sec:transfer}
\vspace{-5pt}

{\color{revcolor}Under Assumption~\ref{asm:right_translation}, we analyze how such right-translation attacks propagate through the dynamics.}

\vspace{-5pt}
\subsection{Dynamical Impact Transfer}\label{sec:inv_transfer}
\vspace{-5pt}
{\color{revcolor}The immediate question is whether a single, fixed $\xi$ remains transferable throughout the trajectory, rather than needing to be recomputed at every step: this requires $d_k=\exp(\xi)$ for every $\exp(f_k)$ encountered along the flow. The next proposition characterizes exactly when this holds for the entire flow generated by a fixed input.}

\begin{prop}\label{prop:invariant_transfer}
    Let $a_k = \exp(\xi_k) \in \Gc$ denote the attack translation induced by the displacement $\xi_k$, via the sensor attack $\tilde{R}_{a_k}h(x_k) = h(x_k a_k)$. The dynamical impact of $\xi_k$ is the same along the entire flow $\exp(sf_k)$, $s \in \Rb$, generated by $f_k$, if and only if $[f_k,\, \xi_k] = 0$.
\end{prop}
\begin{proof}
    The dynamical impact from~\eqref{eq:dyn_impact} is $d_k = \exp(f_k)^{-1}\exp(\xi_k)\exp(f_k) = \exp(\Ad^R_{\exp(f_k)}\xi_k)$. Thus, $d_k$ is independent of $\exp(f_k)$ if and only if $\Ad^R_{\exp(f_k)}\xi_k = \xi_k$.
    
    Consider the curve $\phi(s) := \Ad^R_{\exp(sf_k)}\xi_k$. Its derivative is the linear ODE $\dot{\phi}(s) = -[f_k, \phi(s)]$ with $\phi(0) = \xi_k$~\cite{hall2015lie}. Since $[f_k,\cdot]$ is a bounded linear operator on the finite-dimensional space $\mathfrak{g}$, the Picard-Lindelöf theorem guarantees a unique global solution. Because $[f_k,\xi_k]=0$, the curve $\phi(s) = \xi_k$ satisfies the ODE, yielding $\Ad^R_{\exp(f_k)}\xi_k = \phi(1) = \xi_k$. Conversely, if $\Ad^R_{\exp(f_k)}\xi_k = \xi_k$, the curve is constant, and the derivative at $s=0$ directly yields $[f_k,\xi_k]=0$.
\end{proof}


{\color{revcolor}An attack parametrized by $\xi_k \in \mathfrak{g}$ need only commute with the zero-order hold dynamics $f_k$. Even linear systems obey this restriction: writing the attack as an additive deviation $\delta$ under $x_{k+1}=Ax_k$, transferability requires $\delta = A\delta$, i.e.\ $\delta \in \ker(A-I)$, nonzero only along $A$'s eigenvalue-$1$ (integrator) directions. More generally, on Lie groups, attacks are only invariant across specific steady maneuvers generated by $u_k$.}

Furthermore, we can extend Proposition~\ref{prop:invariant_transfer} to compositions.
\begin{cor}
    Let $a_k = \prod_i \exp(\xi_k^i) \in \Gc$. The dynamical impact of $a_k$ along $\exp(sf_k)$ if $[f_k, \xi_k^i] = 0, \, \forall i$.
\end{cor}
Note that each factor need only commute with $f_k$ individually, so transferable attacks compose freely.

\vspace{-5pt}
\subsection{The Commuting Subalgebra}\label{sec:Invariance}
\vspace{-5pt}
By Proposition~\ref{prop:invariant_transfer}, the displacements commuting with the realized generator $f_k \in \mathfrak{g}$ form its \emph{centralizer}, the Lie-group counterpart to the linear integrator condition $\ker(A-I)$ identified above:
\begin{equation}
    \Delta_{f_k} = \ker([f_k, \cdot]) = \{ \xi \in \mathfrak{g} \mid [f_k, \xi] = 0 \} \subset \mathfrak{g}.
\end{equation}
Distinct inputs $u_k$ can realize different generators $f_k$ that nonetheless share the same centralizer. We write $f_e$ for such an equivalence class of held generators, and $\Delta_{f_e}$ for the centralizer they share. Every $\xi \in \Delta_{f_e}$ therefore transfers invariantly, regardless of the state $x \in \Gc$. An attacker wishing to learn transferable attacks need not know where on the state manifold the system currently is, nor know the exact input $u_k$ that was applied, but only whether the chosen $\xi$ lies in $\Delta_{f_e}$ for the input the system may currently be applying.
\begin{exam}\label{ex:dubins_att_class}
    Consider a Dubins car with $f_e = v e_f {\color{revcolor}+ 0\,e_l} + \omega e_\theta$, where $e_f, e_l$, and $e_\theta$ are the forward, lateral, and heading generators of $SE(2)$. For an attack with the displacement $\xi = a e_f + b e_l + c e_\theta$, the condition $[f_e, \xi] = 0$ expands to:
    \begin{equation*}
        [ve_f + \omega e_\theta,\; ae_f + be_l + ce_\theta] = -\omega b e_f + (\omega a - vc)e_l = 0.
    \end{equation*}
    For straight motion ($\omega = 0, v \neq 0$), this requires $c = 0$, leaving $a$ and $b$ free. Thus, $\Delta_{f_e} = \mathrm{span}(e_f, e_l)$, as illustrated by the purely lateral spoof in the top plot of Figure~\ref{fig:commuting_subspace}. Conversely, for curved motion ($\omega \neq 0$), the condition yields $b = 0$ and $a = \frac{v}{\omega}c$, meaning $\xi = c(\frac{v}{\omega}e_f + e_\theta)$. The only transferable attacks are displacements strictly along the trajectory as seen in the purple arrows in the bottom plot of Figure~\ref{fig:commuting_subspace}. Purely body-consistent (red) or world-consistent (green) offsets fail to commute ($\notin \Delta_{f_e}$) and distort the observed kinematics.
\end{exam}

\begin{figure}[t]
    \centering
    \includegraphics[width=\linewidth]{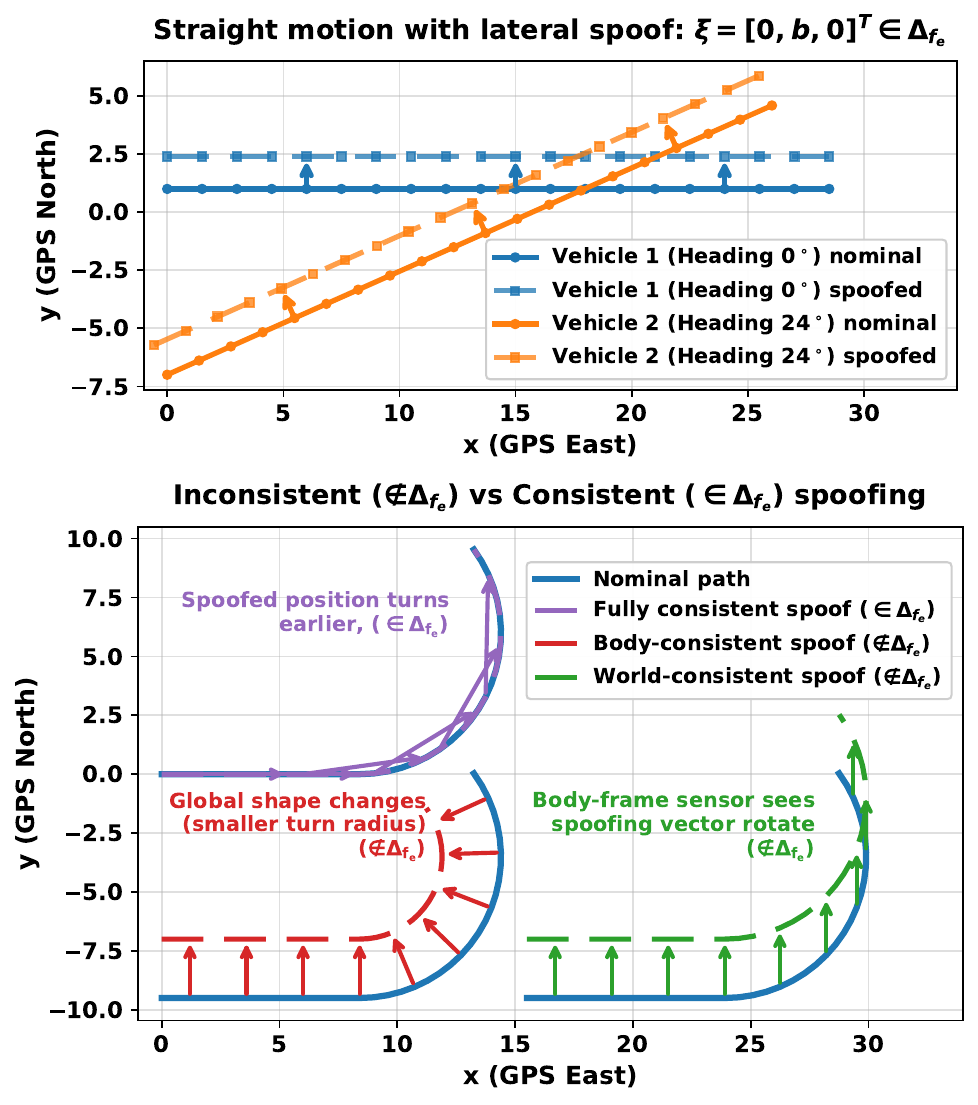}\vspace{-7pt}
    \caption{\small Visualization of the commuting subspace $\Delta_{f_e}$. (Top) A constant lateral attack perfectly commutes during straight-line motion. (Bottom) During turns, only spoofing along the nominal trajectory (purple) remains in $\Delta_{f_e}$. Inconsistent attacks, such as purely body-frame (red) or global (green) offsets, fail to commute and cause observable distortions.} \vspace{-7pt}
    \label{fig:commuting_subspace}
\end{figure}

Note that $\Delta_{f_e}$ is a \emph{Lie subalgebra}, meaning that $[\xi,\eta]\in\Delta_{f_e}$ for all $\xi,\eta\in\Delta_{f_e}$, which follows from the Jacobi identity:
\begin{equation}\label{eq:jacobi}
    [f_e,\,[\xi,\eta]]
    = -[\xi,\,[\eta,f_e]] - [\eta,\,[f_e,\xi]]
    = -[\xi,\,0] - [\eta,\,0] = 0.
\end{equation}
Thus, $\Delta_{f_e}$ is an \emph{involutive distribution}, for which the Frobenius theorem guarantees integrability: restricted to displacements in $\Delta_{f_e}$, the system remains confined to a submanifold, called a \emph{leaf}. Concretely, the leaf through $x$ is the coset $H_x = x \cdot H$ of $H = \exp(\Delta_{f_e})$, coinciding with the global coset whenever $H$ is closed in $\Gc$.

\begin{cor}\label{cor:involutive}
    For all $\xi_0 \in \Delta_{f_e}$, the attacked state $x\cdot\exp(\xi_0)$ remains within the leaf $H_x$ passing through $x$.
\end{cor}
\begin{proof}
    Involutivity~\eqref{eq:jacobi} ensures the flow of $\xi_0$ stays tangent to the manifold $H_x$. Thus, $\exp(\xi_0)$ restricts the state to the local leaf $H_x$, which coincides with the global coset $xH$ if $H$ is a closed subgroup.
\end{proof}

Corollary~\ref{cor:involutive} shows that invariant transferability restricts reachability: to guarantee identical dynamical impact across all states, the attacker is confined to displacements $\xi \in \Delta_{f_e}$, keeping the spoofed state on the leaf $H_x$. Since $H_x$ depends on $u_k$ through $f_k$, a defender can change $u_k$ toward a maneuver whose centralizer $\Delta_{f'_e}$ excludes the attacker's $\xi$, pushing a previously transferable attack off the new leaf.

\vspace{-5pt}
\subsection{Stealth Transfer}\label{sec:obs_transfer}
\vspace{-5pt}
While distorting dynamical impact might mitigate some attacks, a defender's real goal is detection. We therefore seek the analogous stealth condition: does the detector's residual remain unchanged across trajectories? Recall from Definition~\ref{def:stealth} that the equivariant observer framework natively respects the state space geometry by deriving innovations from the invariant error $\rho_k$. Leveraging this structure, we establish the exact algebraic conditions for stealthy transfer.

{\color{revcolor}Recall~\eqref{eq:eta_def} and~\eqref{eq:rho_def}. Combining them yields
\begin{equation*}
    \rho_{k+1}=\exp(-f_k)\exp(-\eta_k)\exp(f_k)\exp(\xi_k).
\end{equation*}
Note that the transferability of the fresh attack, $[\xi_k,f_k]=0$, does not suffice to make this independent of $\exp(f_k)$, since
\begin{equation*}
    \rho_{k+1}=\exp(-f_k)\exp(-\eta_k) \exp(\xi_k) \exp(f_k)
\end{equation*}
still depends on the dynamics $\exp(f_k)$. \textcolor{rev2blue}{Instead, we require an analogous transferability condition on $\eta_k$, derived in the next proposition.}

\begin{prop}\label{prop:stealth_transfer}
    Let $\eta_k\in\mathfrak g$ denote the accumulated detector drift at time $k$. The invariant error $\rho_{k+1}=\exp(-f_k)\exp(-\eta_k)\exp(f_k)\exp(\xi_k)$ is independent of $\exp(f_k)$ along the entire flow $\exp(sf_k)$, $s\in\Rb$, generated by $f_k$, if and only if $[\eta_k, f_k] = 0$.
\end{prop}
\begin{proof}
    Analogous to Proposition~\ref{prop:invariant_transfer}, with $\xi_k$ replaced by $-\eta_k$: $\exp(-f_k)\exp(-\eta_k)\exp(f_k)=\exp(-\eta_k)$ along the entire flow iff $[f_k,\eta_k]=0$, giving $\rho_{k+1}=\exp(-\eta_k)\exp(\xi_k)$, independent of $\exp(f_k)$, so $I_{k+1}^{\textcolor{rev2blue}{\rho}}=h(\rho_{k+1})-h(e)$ depends only on $\exp(-\eta_k)$ and $\xi_k$.
\end{proof}

Proposition~\ref{prop:stealth_transfer}'s condition $[\eta_k,f_k]=0$ is offline-checkable, just like transferability for $\xi_k$. But unlike $\xi_k$, which the attacker chooses directly, $\eta_k$ accumulates via the observer's own correction $K(I_k^\rho)$ in~\eqref{eq:eta_update}, a function the attacker cannot access, making $\eta_k \in \Delta_{f_e}$ far harder to guarantee. This symmetrically hands the defender a lever: by varying $u_k$ so that $f_k$ no longer commutes with $\eta_k$, the defender exposes attacks that were stealthy only on the nominal trajectory.}

\section{Main Results}\label{sec:main_results}
\vspace{-5pt}
The previous section characterized ideal transfer conditions. In practice, the true dynamics $f_k$ is never perfectly known, so an attacker cannot explicitly compute the invariant subspace $\Delta_{f_e}$. We decompose the \textcolor{rev2blue}{empirical} attack $\hat{\xi}_k$ and detector drift $\hat{\eta}_k$ onto the true, unknown subspaces:
\begin{equation}\label{eq:decomp}
    \hat{\xi}_k = \underbrace{\xi_k}_{\in \Delta_{f_e}} + \underbrace{\rho^\xi_k}_{\not \in \Delta_{f_e}}, \quad
    \hat{\eta}_k = \underbrace{\eta_k}_{\in \Delta_{f_e}} + \underbrace{\rho^\eta_k}_{\not \in \Delta_{f_e}},
\end{equation}
\textcolor{rev2blue}{Here, $\xi_k, \eta_k \in \Delta_{f_e}$ are the ideal transferable components.} The bracket-violating residuals $\rho^\xi_k, \rho^\eta_k$ mathematically capture the attacker's fundamental ignorance of the true dynamics, alongside sensor noise and trajectory mismatches.

\begin{defi}\label{def:rich}
    An attack dataset $\Dc^{a}_{x_0^D}$ is \emph{$\epsilon$-rich} with respect to $\Delta_{f_e}$ if its empirical samples yield components such that:
    \begin{enumerate}
        \item The ideal components $\{\xi_k^{a,(i)}\}_{i=1}^M$ span $\Delta_{f_e}\subset\mathfrak{g}$, and
        \item The residual errors satisfy $\|\rho^{\xi,(i)}_k\|\leq\epsilon$ for all $i,k$.
    \end{enumerate}
\end{defi}

Intuitively, an $\epsilon$-rich dataset guarantees that the attacker has empirically collected a sufficiently diverse set of noisy base attacks to implicitly span the required geometric degrees of freedom, while ensuring that the residual $\rho^\xi_k$ caused by their ignorance of $\Delta_{f_e}$ is strictly bounded by $\epsilon$.

With these residual errors defined, we now evaluate the change in impact and stealth when transferring an attack.
\begin{thm}\label{thm:main}
    Let $\Dc^{a}_{x_0^D}$ be an $\epsilon$-rich attack dataset. Suppose an attacker extracts a sensor attack {\color{revcolor}$\tilde{R}_{a_k}$} where $a_k = \exp(\hat{\xi}_k)$, and applies it to the victim's measurements such that $\tilde{y}_k = {\color{revcolor}\tilde{R}_{a_k}}h(x_k) = h(x_k a_k)$. Then, the transferred attack exhibits the following properties:
    \begin{enumerate}
        \item The \emph{dynamical impact} distortion of the realized attack $\hat{\xi}_k$ satisfies
              \begin{equation} \label{eq:dyn_impact_error}
                  d_k = \exp(\xi_k + \Ad^R_{\exp(f_k)} \rho^\xi_k),
              \end{equation}
              with the bound $\| \Ad^R_{\exp(f_k)} \rho^\xi_k \| \leq \epsilon \| \Ad^R_{\exp(f_k)}\|$.

        \item Under an equivariant observer, the \emph{stealth} distortion is reflected in the realized invariant error $\rho_{k+1}$, which evaluates to:
              \begin{equation} \label{eq:stealth_exact_diff}
                  \rho_{k+1} = \exp\bigl(-(\eta_{k} + \Ad^R_{\exp(f_k)} \rho^\eta_{k})\bigr) \exp(\xi_k + \rho^\xi_k),
              \end{equation}
              yielding the realized innovation $I_{k+1}^{\textcolor{rev2blue}{\rho}} = h(\rho_{k+1}) - h(e)$.
    \end{enumerate}
\end{thm}
\begin{proof}
    We prove the two cases separately.

    1) For the \emph{dynamical impact} distortion, by~\eqref{eq:dyn_impact},
    \begin{equation*}
        d_k = \exp(-f_k)\exp(\hat{\xi}_k)\exp(f_k) = \exp(\Ad^R_{\exp(f_k)}\hat{\xi}_k).
    \end{equation*}
    Decomposing $\hat{\xi}_k = \xi_k + \rho^\xi_k$ and using linearity of $\Ad$, $d_k = \exp(\Ad^R_{\exp(f_k)}\xi_k + \Ad^R_{\exp(f_k)}\rho^\xi_k)$; since $\xi_k \in \Delta_{f_e}$, $\Ad^R_{\exp(f_k)}\xi_k = \xi_k$, giving~\eqref{eq:dyn_impact_error}.

2) For the \emph{stealth} distortion, the innovation is evaluated on the invariant error $\rho_{k+1} = \hat{x}_{k+1|k}^{-1}x_{k+1}$, so $I_{k+1}^{\textcolor{rev2blue}{\rho}} = h(\rho_{k+1}) - h(e)$. Substituting $x_{k+1} = x_{k}\exp(f_k)\exp(\hat{\xi}_k)$ and $\hat{x}_{k+1|k} = x_{k}\exp(\hat{\eta}_{k})\exp(f_k)$, yields:
    \begin{equation*}
    \begin{aligned}
        \rho_{k+1} &= \bigl(x_{k}\exp(\hat{\eta}_{k})\exp(f_k)\bigr)^{-1} \bigl(x_{k}\exp(f_k)\exp(\hat{\xi}_k)\bigr) \\
        &= \exp(-\Ad^R_{\exp(f_k)}\hat{\eta}_{k})\exp(\hat{\xi}_k).
    \end{aligned}
    \end{equation*}
    With $\hat{\xi}_k = \xi_k + \rho^\xi_k$, $\hat{\eta}_{k} = \eta_{k} + \rho^\eta_{k}$, the linearity of $\Ad$, and $\Ad^R_{\exp(f_k)}\eta_{k} = \eta_{k}$ since $\eta_{k} \in \Delta_{f_e}$, we obtain~\eqref{eq:stealth_exact_diff}.
\end{proof}

Equations \eqref{eq:dyn_impact_error} and \eqref{eq:stealth_exact_diff} show that without $\rho^\xi_k, \rho^\eta_k$, the attack transfers perfectly. Otherwise, dynamical impact distortion is bounded by $\epsilon\|\Ad^R_{\exp(f_k)}\|$ and stealth distortion gains drift via $\Ad^R_{\exp(f_k)}\rho^\eta_{k}$. Equation \eqref{eq:stealth_exact_diff} also reveals a loophole: an attack with nonzero $\rho^\xi_k, \rho^\eta_k$ can still stay stealthy if $h(\rho_{k+1}) = h(\tilde{\rho}_{k+1})$, where $\tilde{\rho}_{k+1} \coloneqq \exp(-\eta_{k})\exp(\xi_k)$ is the ideal error. This mimics the linear case where residual errors in $\rho_{k+1}$ map to the same sensory footprint as the ideal error.

\begin{rem}
    \textcolor{rev2blue}{Theorem \ref{thm:main} exposes an asymmetry: $\rho^\xi_k$ scales the dynamical impact via $\Ad^R_{\exp(f_k)}$, but enters the stealth distortion unscaled; instead it is $\rho^\eta_k$ that picks up the Adjoint. This suggests a strategy for stealth robustness, if the attacker can also shape $\xi_k$ against a prior over candidate $\exp(f_k)$ to keep $\Ad^R_{\exp(f_k)}\rho^\eta_k$ small.}
\end{rem}

\vspace{-5pt}
\section{Numerical Example}\label{sec:examples}
\vspace{-5pt}
Consider a Dubins car with state $z = (x, y, \theta) \in SE(2)$, whose continuous-time dynamics in global coordinates are 
\begin{equation} 
    \dot x = v\cos\theta, \quad \dot y = v\sin\theta, \quad \dot\theta = \omega, \vspace{-5pt}
\end{equation}
or equivalently in the body frame, $\dot z = z(ve_f + \omega e_\theta)$, so $f_e = ve_f + \omega e_\theta \in \mathfrak{se}(2)$. The vehicle measures its position via a GPS-like sensor and a LIDAR-like odometry sensor, yielding the mixed sensor suite $h(z) = [x, y, f_s, l_s]^\top$.

Let the nominal training trajectory move straight along the global $x$-axis ($\theta = 0$, $\omega = 0$), expanding the commuting subspace to $\Delta_{f_e} = \mathrm{span}(e_f, e_l)$ (Example~\ref{ex:dubins_att_class}). Starting at $t = 60$\,s, the attacker observes a multi-modal lateral deviation $\Delta h = [0, 10, 0, 10]^\top$. Because the global and body frames align when $\theta = 0$, this combined sensor drift maps directly to the right-invariant local displacement $\xi_k = [0,10,0]^\top \in \mathfrak{se}(2)$.

To deploy this learned displacement against a victim on an arbitrary trajectory, simply replaying the constant bias $\Delta h$ would violate physical consistency between the global and local sensors. Instead, the attacker dynamically coordinates the spoofing vector using their heading estimate $\hat{\theta}$:
\begin{equation}
    h(z^V_{\text{attacked}}) = h(z^V) + \begin{bmatrix} -10 \sin \hat{\theta} & 10 \cos \hat{\theta} & 0 & 10 \end{bmatrix}^\top. \vspace{-2pt}
\end{equation}
This orientation-dependent update successfully realizes the multi-sensor observation action $\tilde R_{a_k} h(z^V) = h(z^V\exp(\xi_k))$.

By Proposition~\ref{prop:invariant_transfer} and Theorem~\ref{thm:main}, transferability and impact are governed by $\Delta_{f_e}$ and the Adjoint operator $\Ad^R_{\exp(f_k)}$ of the one-step flow, with $f_k \coloneqq f_e\Delta t$. For any Dubins maneuver, along-trajectory displacements lie in $\Delta_{f_e}$, so as Figure~\ref{fig:adjoint_error_bound} shows, this invariant attack smoothly ``drags'' the detector bias forward along the predicted path.

In practice, synthesized attacks may introduce out-of-subspace residuals $\rho^\xi \notin \Delta_{f_e}$. Their amplification is dictated by the Adjoint operator of the local flow $\exp(f_k) = (p_f, p_l, p_\theta) \in SE(2)$:
\vspace{-10pt}
\begin{equation}\label{eq:adjoint_general}
\Ad^R_{\exp(f_k)} = \begin{bmatrix} \cos p_\theta & -\sin p_\theta & p_l \\ \sin p_\theta & \cos p_\theta & -p_f \\ 0 & 0 & 1 \end{bmatrix},
\end{equation}
with the induced 2-operator norm yielding $\|\Ad^R_{\exp(f_k)}\|_2 = \frac{1}{2}\left( r + \sqrt{r^2 + 4} \right)$, where $r = \sqrt{p_f^2 + p_l^2}$ is the translational magnitude. This directly provides the dynamical impact bound $\|\Ad^R_{\exp(f_k)}\rho^\xi\|_2 \leq \epsilon\|\Ad^R_{\exp(f_k)}\|_2$ from Theorem~\ref{thm:main}. Mismatches in the heading direction excite the last column, scaled by the spatial displacements $p_f$ and $p_l$, as bounded by the dashed circles in Figure~\ref{fig:adjoint_error_bound}.

\begin{figure}[t]
    \centering
    \includegraphics[width=0.90\linewidth]{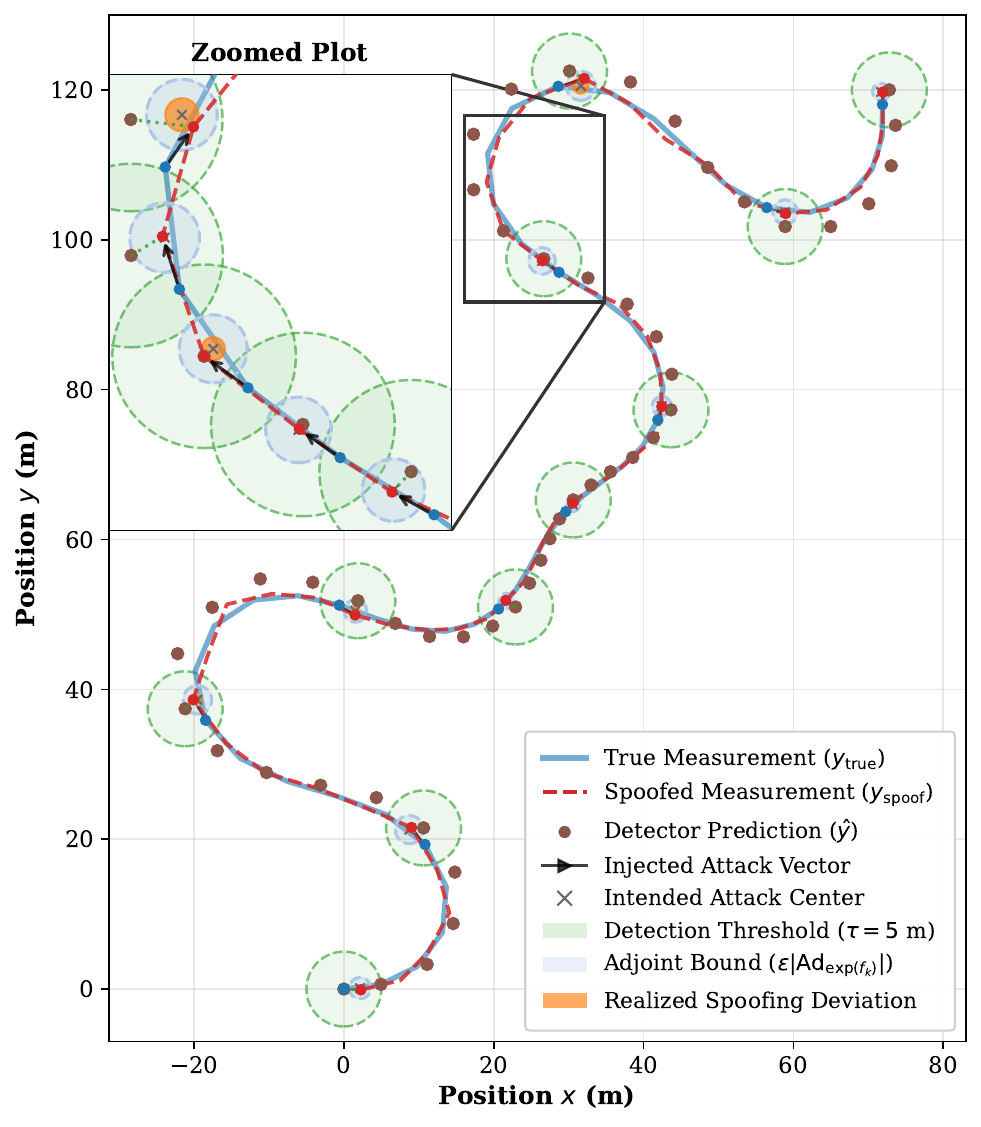}\vspace{-9pt}
    \caption{\small Estimator drift under the spoofed attack: realized
    deviations (orange) stay within the Adjoint bound (dashed blue) and
    below $\tau = 5.0$\,m (dashed green). Zoomed plot: the injected residual is
    purely lateral, so the Adjoint only rotates it~\eqref{eq:adjoint_error},
    leaving its magnitude at $\epsilon$ despite large $p_f, p_l$.} \vspace{-7pt}
    \label{fig:adjoint_error_bound}
\end{figure}

Crucially, for $\rho^\xi = [0,\epsilon,0]^\top$, the last column is not excited:
\begin{equation}\label{eq:adjoint_error}
    \Ad^R_{\exp(f_k)}\rho^\xi = \begin{bmatrix} -\epsilon\sin p_\theta & \epsilon\cos p_\theta & 0 \end{bmatrix}^\top.
\end{equation}
These zero entries bypass $p_f$ and $p_l$ entirely: the residual merely rotates, preserving its exact magnitude $\epsilon$.

Consider the simulated victim at $t = 20.0$\,s, where $v = 13.96$\,m/s and $\omega = -1.02$\,rad/s over $\Delta t = 0.5$\,s, yielding local integration parameters $p_f = 6.680$\,m, $p_l = -1.742$\,m, $p_\theta = -0.510$\,rad and $\|\Ad^R_{\exp(f_k)}\|_2 = 7.045$. The attacker targets $\xi = \alpha f_e = [3.015,\, 0,\, -0.220]^\top \in \Delta_{f_e}$ with $\alpha = 0.216$, and injects a lateral residual $\rho^\xi = [0,\epsilon,0]^\top$ with $\epsilon = 0.277$. By~\eqref{eq:adjoint_error}, the residual rotates under the Adjoint, giving $\|\Ad^R_{\exp(f_k)}\rho^\xi\|_2 = \epsilon = 0.277 \leq \epsilon\|\Ad^R_{\exp(f_k)}\|_2 = 1.952$, yielding
\vspace{-0pt}
\begin{equation}
    \xi + \Ad^R_{\exp(f_k)}\rho^\xi
    = \begin{bmatrix} 3.015 \\ 0 \\ -0.220 \end{bmatrix}
    + \begin{bmatrix} 0.135 \\ 0.242 \\ 0 \end{bmatrix}
    = \begin{bmatrix} 3.150 \\ 0.242 \\ -0.220 \end{bmatrix}, \vspace{-5pt}
\end{equation}
which stays close to $\exp(\xi)$ despite the large $p_f, p_l$.

Figure~\ref{fig:training_vs_transfer} isolates the structural design from noise over time. During training (dotted), the detector innovation stays below the dynamical impact, since the estimator tracks the spoofed signal so faithfully that its residual understates the true displacement. After transfer with $\epsilon=0.277$ (solid), accumulated noise lets the innovation transiently exceed the impact, yet both remain within the bound $\|\xi\| + \epsilon\|\Ad^R_{\exp(f_k)}\|_2$ (dashed blue, peaking at $4.99$\,m, just below $\tau=5.0$\,m), confirming the training margin maintains $\tau$-stealthiness at deployment.


\vspace{-5pt}
\section{Conclusions}\label{sec:conclusions}
\vspace{-5pt}
This paper introduces a geometric framework for sensor spoofing transferability on Lie groups, showing that transferability reduces to a single algebraic object: the centralizer $\Delta_{f_e} = \ker([f_e, \cdot])$ of the nominal flow. \textcolor{rev2blue}{Unlike linear systems, where persistence reduces to a single integrator condition, non-commutative dynamics generalize this invariance to $\Delta_{f_e}$.} Imperfect attacks reveal a structural asymmetry: physical trajectory deviations are amplified by the Adjoint action, while detector innovations are invariant to the observation map, as validated on the $SE(2)$ Dubins unicycle. The centralizer $\Delta_{f_e}$ thus reframes attack transferability as a kinematic constraint intrinsic to the system's symmetry.

\begin{figure}[t]
    \centering
    \includegraphics[width=\linewidth]{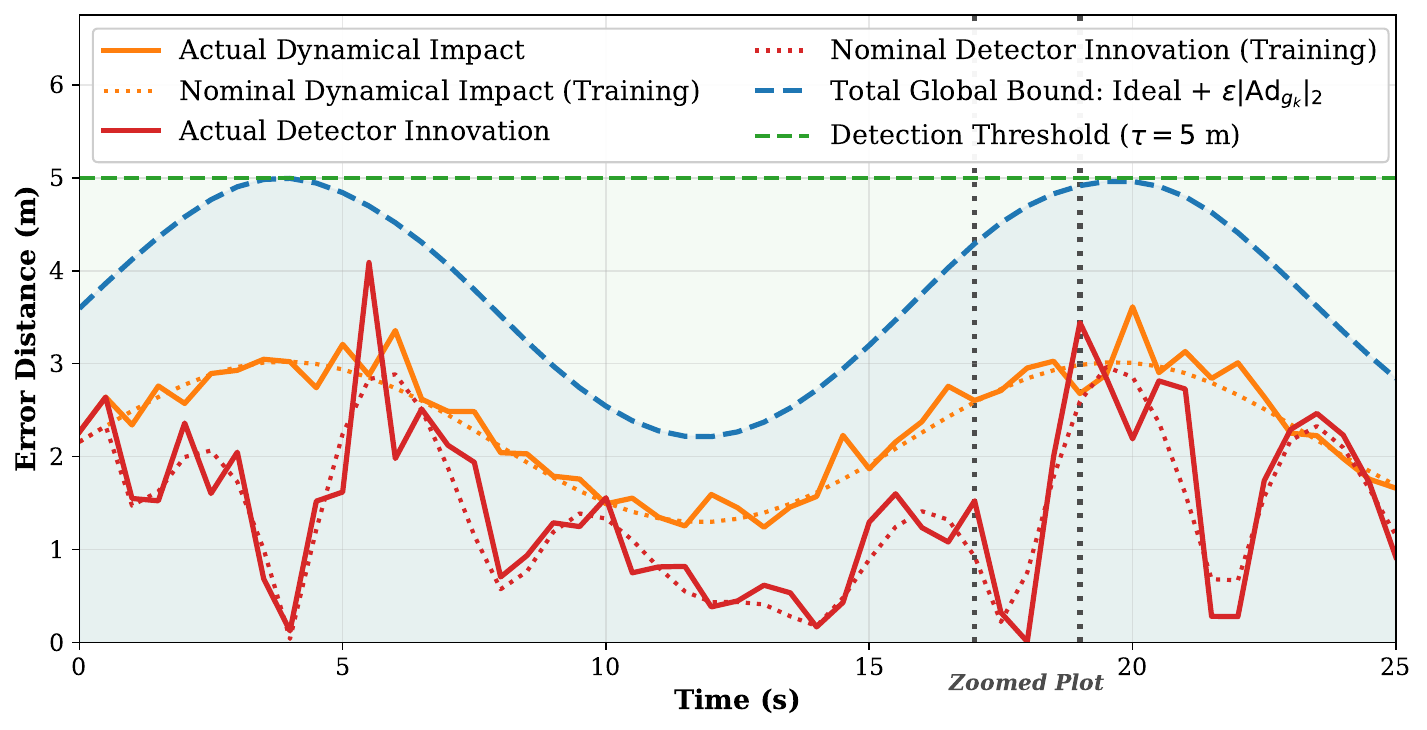}\vspace{-9pt}
    \caption{\small Impact and innovation stay near the nominal. However with error $\epsilon = 0.277$ (solid), both grow but
    remain within $\|\xi\| + \epsilon\|\Ad^R_{\exp(f_k)}\|_2$ (dashed blue),
    confirming $\tau$-stealthiness.} \vspace{-7pt}
    \label{fig:training_vs_transfer}
\end{figure}

Future work will extend this framework to closed-loop systems under coordinated actuator-sensor spoofing, tighten the Adjoint bounds and $\epsilon$-conditions for $\tau$-stealth, and generalize the invariance assumptions to direct-measurement observers and time-varying inter-sample inputs.

\vspace{-7pt}
\bibliographystyle{unsrt}
\bibliography{bibliography.bib}

\begin{thebibliography}{10}

\bibitem{tippenhauer2011requirements}
N.~O. Tippenhauer, C.~P\"{o}pper, K.~B. Rasmussen, and S.~Capkun.
\newblock On the requirements for successful {GPS} spoofing attacks.
\newblock In {\em Proc. 18th {ACM} Conf. Comput. Commun. Secur.}, CCS '11,
  pages 75--86, New York, NY, USA, 2011. ACM.

\bibitem{kerns2014unmanned}
A.~J. Kerns, D.~P. Shepard, J.~A. Bhatti, and T.~E. Humphreys.
\newblock Unmanned aircraft capture and control via {GPS} spoofing.
\newblock {\em J. Field Robotics}, 31(4):617--636, 2014.

\bibitem{lyons2026distributed}
L.~Lyons, M.~Boldrer, and L.~Ferranti.
\newblock Distributed attack-resilient platooning against false data injection.
\newblock {\em IEEE Trans. Veh. Technol.}, 75(3):3888--3903, 2026.

\bibitem{amin2013cyber}
S.~Amin, X.~Litrico, S.~Sastry, and A.~M. Bayen.
\newblock Cyber security of water {SCADA} systems—{Part I}: Analysis and
  experimentation of stealthy deception attacks.
\newblock {\em IEEE Trans. Control Syst. Technol.}, 21(5):1963--1970, 2013.

\bibitem{mo2009replay}
Y.~Mo and B.~Sinopoli.
\newblock Secure control against replay attacks.
\newblock In {\em 2009 47th Annual Allerton Conf. Commun., Control, Comput.
  ({Allerton})}, pages 911--918, 2009.

\bibitem{pasqualetti2013attack}
F.~Pasqualetti, F.~D\"{o}rfler, and F.~Bullo.
\newblock Attack detection and identification in cyber-physical systems.
\newblock {\em IEEE Trans. Autom. Control}, 58(11):2715--2729, 2013.

\bibitem{chong2019tutorial}
M.~S. Chong, H.~Sandberg, and A.~Teixeira.
\newblock A tutorial introduction to security and privacy for cyber-physical
  systems.
\newblock In {\em 2019 18th Eur. Control Conf. ({ECC})}, pages 968--978, 2019.

\bibitem{bloch2015nonholonomic}
A.~M. Bloch.
\newblock {\em Nonholonomic Mechanics and Control}.
\newblock Springer, New York, NY, 2nd edition, 2015.

\bibitem{zhang2022stealthy}
K.~Zhang, C.~Keliris, T.~Parisini, and M.~M. Polycarpou.
\newblock Stealthy integrity attacks for a class of nonlinear cyber-physical
  systems.
\newblock {\em IEEE Trans. Autom. Control}, 67(12):6723--6730, 2022.

\bibitem{kwon2024perfectly}
H.~Kwon, H.~Kawase, H.~A. Nieves-Vazquez, K.~Kogiso, and J.~Ueda.
\newblock Perfectly undetectable false data injection attacks on encrypted
  bilateral teleoperation system based on dynamic symmetry and malleability.
\newblock In {\em 2025 IEEE Int. Conf. Robotics Automation (ICRA)}, pages
  16457--16463, 2025.

\bibitem{murray1994mathematical}
R.~M. Murray, Z.~Li, and S.~S. Sastry.
\newblock {\em A Mathematical Introduction to Robotic Manipulation}.
\newblock CRC Press, 1st edition, 1994.

\bibitem{barrau2017invariant}
A.~Barrau and S.~Bonnabel.
\newblock The invariant extended {Kalman} filter as a stable observer.
\newblock {\em IEEE Trans. Autom. Control}, 62(4):1797--1812, 2017.

\bibitem{mahony2022observer}
R.~Mahony, P.~van Goor, and T.~Hamel.
\newblock Observer design for nonlinear systems with equivariance.
\newblock {\em Annu. Rev. Control Robot. Auton. Syst.}, 5:221--252, 2022.

\bibitem{coulson2019deepc}
J.~Coulson, J.~Lygeros, and F.~D\"{o}rfler.
\newblock Data-enabled predictive control: In the shallows of the {DeePC}.
\newblock In {\em 2019 18th Eur. Control Conf. ({ECC})}, pages 307--312, 2019.

\bibitem{berberich2020nonlinear}
J.~Berberich, J.~K\"{o}hler, M.~A. M\"{u}ller, and F.~Allg\"{o}wer.
\newblock Linear tracking {MPC} for nonlinear systems---{Part II}: The
  data-driven case.
\newblock {\em IEEE Trans. Autom. Control}, 67(9):4406--4421, 2022.

\bibitem{alisic2021ecc}
R.~Alisic and H.~Sandberg.
\newblock Data-injection attacks using historical inputs and outputs.
\newblock In {\em 2021 Eur. Control Conf. ({ECC})}, pages 1399--1405, 2021.

\bibitem{taheri2021data}
M.~Taheri, K.~Khorasani, I.~Shames, and N.~Meskin.
\newblock Data-driven covert-attack strategies and countermeasures for
  cyber-physical systems.
\newblock In {\em 2021 60th IEEE Conf. Decis. Control ({CDC})}, pages
  4170--4175, 2021.

\bibitem{hall2015lie}
B.~C. Hall.
\newblock {\em Lie Groups, Lie Algebras, and Representations: An Elementary
  Introduction}, volume 222 of {\em Graduate Texts in Mathematics}.
\newblock Springer, 2nd edition, 2015.

\end{thebibliography}

\end{document}